\documentclass[
journal]{IEEEtran}                                      
\usepackage{rotating}
\usepackage{amsmath}
\ifCLASSINFOpdf
\usepackage{graphicx,amssymb}
\usepackage{algorithm}
 \usepackage{algorithmic}
 \usepackage{diagbox}
\usepackage{multirow}
\usepackage{subcaption}
\usepackage{makecell}

\makeatletter

\@addtoreset{algorithm}{section}
\makeatother

\DeclareGraphicsExtensions{.pdf,.jpg,.jpeg,.png,.eps}
\else
\usepackage[dvips]{graphicx}
\usepackage{epsfig}
\usepackage{epstopdf}
\DeclareGraphicsExtensions{.eps}
\fi

\usepackage{amsthm}

\newtheorem{theorem}{Theorem}[section]

\newtheorem{remark}[theorem]{Remark}

\numberwithin{equation}{section}
\usepackage{color}

\usepackage{soul}

\begin{document}
\title{  Preconditioned Gradient Descent Algorithm for Inverse Filtering on  Spatially Distributed Networks}
\author{Cheng Cheng, Nazar Emirov, and Qiyu Sun 
\thanks{Cheng is with the Department of Mathematics, Duke University,   Durham, NC 27708; Nazar and Sun is with the  Department of Mathematics, University of Central Florida, Orlando, Florida 32816.
Emails:  cheng87@math.duke.edu;  nazaremirov@knights.ucf.edu;  qiyu.sun@ucf.edu. This work is partially supported by 
Simons Math+X Investigators Award (400837) and the  National Science Foundation (DMS-1816313).
}
}

\maketitle

\begin{abstract}
Graph filters and their inverses have been widely used in 
	denoising, smoothing, sampling, interpolating and learning.  
 Implementation of  an inverse filtering procedure
   on spatially distributed networks  (SDNs) 
   is a remarkable challenge, as
  each agent on an SDN is  equipped  with a data processing subsystem with limited capacity and
a communication subsystem with confined range due to engineering limitations. 
In this letter, we introduce a preconditioned gradient descent algorithm to implement the inverse filtering procedure associated with a  graph filter having small geodesic-width. The proposed algorithm converges exponentially, and
it  can be implemented at vertex level
  and   applied  to  time-varying inverse filtering  on  SDNs. 
\end{abstract}


\vskip-1.8mm

\section{Introduction}

Spatially distributed networks (SDNs) have been widely used in (wireless) sensor networks,
drone fleets, smart grids
 and many real world applications 
\cite{Yick08}--\cite{Cheng17}.
 An  SDN has a large amount of agents and each agent
 equipped  with  a data processing subsystem having
limited  data storage and computation power
    and a communication subsystem  for  data exchanging  to its ``neighboring" agents
  within  communication 
   range. 
The topology of an SDN can be described by a connected,  undirected  and unweighted finite graph
${\mathcal G}:=(V, E)$  with a vertex  in  $V$ 
 representing an agent and an edge in $E$ between vertices  indicating that
the corresponding agents are within some range in the spatial space.
 In this letter,
 we consider SDNs equipped with a  communication subsystem  at each agent   to directly communicate between two agents
 if the  geodesic distance  
 between their corresponding vertices $i,j\in V$
 is at most $L$,
 \vspace{-.5em}
 \begin{equation}\label{communicationreange.def}
 \rho(i,j)\le L,
 \vspace{-.5em}
 \end{equation}
 where  the geodesic distance
 $\rho(i,j)$ is  the  number of edges in a shortest path connecting  $i, j\in V$,  and we call  the  minimal integer $L\ge 1$ in \eqref{communicationreange.def}
 as  the {\em communication range} of the SDN.
    Therefore
 the implementation of data processing on our SDNs is a distributed 
 task and
 it should be designed at agent/vertex level with confined communication range.  In this letter, we consider the implementation of graph filtering and inverse filtering on SDNs, which are required to be fulfilled at agent level with  communication range  no more than $L$.

 A  signal  on a graph ${\mathcal G}=(V, E)$ is a vector ${\bf x}=(x(i))_{i\in V}$ indexed by the vertex set,
  and a  graph filter ${\bf H}$ maps a  graph signal  ${\bf x}$ 
 linearly
   to another graph  signal  ${\bf y}={\bf H}{\bf x}$,   
 which
is usually represented by a matrix ${\bf H} = (H(i,j))_{i,j\in V}$ indexed by vertices in $V$.
 Graph filtering ${\bf x} \mapsto {\bf Hx}$ and its inverse filtering
 ${\bf y}\mapsto {\bf H}^{-1} {\bf y}$ play important roles in graph signal processing
   and they have been used in
    smoothing, sampling, interpolating and many real-world applications
   \cite{shuman13}, \cite{aliaksei13}--\cite{jiang2021}.
 A graph
 filter ${\bf H} = (H(i,j))_{i,j\in V}$ is said to have {\em geodesic-width} $\omega({\bf H})$   if
\vspace{-.5em}\begin{equation}\label{bandwidth.def}
H(i,j)=0\ \ {\rm  for\ all} \ \ i, j\in V \ \ {\rm with} \ \  \rho(i,j)>\omega({\bf H}) \vspace{-.5em}\end{equation}
\cite{Cheng17,  Tay19, jiang19}.  
For a 
 filter ${\bf H} = (H(i,j))_{i,j\in V}$ with geodesic-width
 $\omega({\bf H})$, 
 the  corresponding
filtering process
\vspace{-.5em}
\begin{equation}\label{filterprocedure.def}
(x(i))_{i\in V}=:{\bf x} \longmapsto {\bf Hx} ={\bf y}:= (y(i))_{i\in V} \vspace{-.5em}\end{equation}
 can be implemented
at vertex level, and
 the output  at a vertex $i\in V$ is a ``weighted" sum of  the input  in its $\omega({\bf H})$-neighborhood,
\vspace{-.5em}\begin{equation}\label{inputoutput.localimplementation} y(i)=\sum_{\rho(j, i)\le \omega({\bf H})} H(i,j)x(j). \vspace{-.5em} \end{equation}
For  SDNs with communication range $L\ge \omega({\bf H})$, the above implementation at vertex level
 provides an essential tool
	for the filtering procedure \eqref{filterprocedure.def}, in which each agent  $i\in V$  has equipped with subsystems to store
 $H(i,j)$ and $x(j)$ with $\rho(j, i)\le \omega({\bf H})$,
to compute addition and multiplication in \eqref{inputoutput.localimplementation}, and
to exchange   data  to its neighboring agents  $j\in V$ satisfying $\rho(j,i)\le \omega({\bf H})$.

For an invertible filter ${\bf H}$, the  implementation of the inverse filtering procedure
\vspace{-.8em}
\begin{equation}\label{inversefilterprocedure.def}
{\bf y} \longmapsto {\bf H}^{-1}{\bf y}=:{\bf x} \vspace{-.3em}\end{equation}
cannot be directly applied for our SDNs, 
since the  inverse
  filter  ${\bf H}^{-1}$ may have geodesic-width {\em larger} than the communication range $L$. 
  For the consideration of implementing inverse filtering on an SDN with communication range $L\ge 1$, we construct a diagonal preconditioning matrix  ${\bf P}_{\bf H}$  in \eqref{da.def} at vertex level,
and propose
the  preconditioned gradient descent algorithm (PGDA) \eqref{gradientdescent.al2.default}
to implement inverse filtering on the SDN, see Algorithms \ref{preconditioningmatrix.algorithm} and \ref{distributed_ICPA.algorithm}.

 A conventional approach to implement the  inverse filtering  procedure
 \eqref{inversefilterprocedure.def}
is
via the iterative quasi-Newton method
\vspace{-.3em} \begin{equation}\label{Approximationalgorithm}
 {\bf e}^{(m)}= {\bf H}{\bf x}^{(m-1)}-{\bf y}\ {\rm and} \
{\bf x}^{(m)}={\bf x}^{(m-1)}-{\bf G}{\bf e}^{(m)}, \ m\ge 1,
\vspace{-.3em}\end{equation}
with arbitrary initial ${\bf x}^{(0)}$, where the graph filter ${\bf G}$ is an approximation to the inverse ${\bf H}^{-1}$.
A challenge in the quasi-Newton method 
  is how to select the approximation filter ${\mathbf G}$  appropriately.
For   the  
 widely used polynomial graph filters
 $ {\bf H}=h({\bf S})=\sum_{k=0}^{K}h_{k}{\bf S}^k $
 of a graph shift ${\bf S}$ where $h(t)=\sum_{k=0}^{K}h_{k}t^k$ \cite{jiang19}--\cite{Emirov20},
 several methods  
 have been proposed
to construct
  polynomial approximation filters
 ${\bf G}$
 \cite{
 Leus17, isufi19, Shuman18, Emirov19,   Emirov20}.
  However, for the convergence of the  corresponding quasi-Newton method, 
    some prior knowledge is required for  the polynomial $h$ and
      the graph shift ${\bf S}$,
  such as the whole spectrum  
   of the  
   shift
    ${\bf S}$ in the optimal polynomial approximation method
  \cite{Emirov20},
  the interval   containing the spectrum of  the   shift
   ${\bf S}$ 
   in the Chebyshev approximation method
  \cite{Shuman18, Emirov19, Emirov20}, and
  the spectral radius  
  of the  shift
  ${\bf S}$
   and  the zero set of the polynomial $h$
  in the  autoregressive moving average
 filtering  algorithm
 \cite{Leus17, isufi19}.
  For a  non-polynomial graph filter  ${\bf H}$, the approximation filter 
  in the
  gradient descent method  
  is of the form
$ {\bf G}=\beta {\bf H}^T$ 
  with  selection  of
 the  optimal step length $\beta$ depending on maximal and minimal  singular values of the filter  ${\bf H}$  \cite{ sihengTV15, Shi15}, and
 the approximation filter 
 in the iterative matrix inverse approximation  algorithm (IMIA)   could be selected
 under  a strong assumption on   ${\bf H}$ \cite[Theorem 3.2]{Tay19}.
  The proposed PGDA  \eqref{gradientdescent.al2.default}
is the quasi-Newton method \eqref{Approximationalgorithm} with ${\bf P}_{\bf H}^{-2} {\bf H}^T$ being selected as
 the approximation filter ${\bf G}$, see \eqref{hht.singularvalue}.  Comparing with the quasi-Newton methods in   \cite{Tay19, sihengTV15,
 Leus17,  isufi19, Shuman18, Emirov19,   Emirov20, Shi15}, 
one significance of the proposed PGDA is that
the sequence
${\bf x}^{(m)}, m\ge 0$,  in \eqref{gradientdescent.al2.default}
converges exponentially to the output  ${\bf x}$ of the inverse filtering procedure \eqref{inversefilterprocedure.def} whenever the filter ${\bf H}$ is invertible,
see Theorems \ref{exponentialconvergence.tm} and \ref{symmetricexponentialconvergence.tm}.

 For a {\em time-varying} filter 
 the quasi-Newton method \eqref{Approximationalgorithm} to implement their inverse filtering on  SDNs should be {\em self-adaptive} as each agent  
   does not have the whole time-varying filter   
   and it only receives
  the submatrix 
   of the filter  
   within the communication range \cite{Cheng17}.
  The IMIA  algorithm  is  self-adaptive  \cite[Eq. (3.4)]{Tay19},
 and   the
   gradient descent method  \cite{ sihengTV15, Shi15} is not self-adaptive in general  except that the step
length  $\beta$ can be chosen to be {\em time-independent}.
 The second significance  of the proposed PGDA is its {\em self-adaptivity} and  {\em  compatibility} 
 to implement the time-varying inverse filtering procedure  on our SDNs, 
 as  the  preconditioner ${\bf P}_{\bf H}$ (and hence the approximation filter ${\bf P}_{\bf H}^{-2} {\bf H}^T$ in the PGDA)
 is constructed at  vertex level with confined communication range,
  see  Algorithm \ref{preconditioningmatrix.algorithm}.

\vspace{-.05in}
\section{Preconditioned gradient descent algorithm for inverse filtering}\label{sec:pre}
\vspace{-.02in}

Let ${\mathcal G}:=(V, E)$ be a   connected, undirected and unweighted graph and ${\bf H}=(H(i,j))_{i,j\in V}$ be a  filter on the graph  ${\mathcal G}$   with  geodesic-width  $\omega({\bf H})$.
   In this section, we induce a diagonal matrix ${\bf P}_{\bf H}$
 with  diagonal elements  $P_{\bf H}(i, i), i\in V$, given  by
 \vspace{-.3em}\begin{eqnarray} \label{da.def}
&\hskip-0.08in  { P}_{\bf H}(i, i)  
  := & \hskip-0.08in    \max_{k\in B(i, \omega({\bf H}))} \Big\{
 \max \Big( \sum_{j\in B(k, \omega({\bf H}))} |H(j, k)|, \nonumber\\
& \hskip.08in&\qquad\qquad\quad
\sum_{j\in B(k, \omega({\bf H}))} |H(k, j)|\Big)\Big\},
 \vspace{-.3em}\end{eqnarray}
 where we denote
 the set of all $s$-hop neighbors of a vertex $i\in V$  by $B(i,  s)=\{j\in V, \ \rho(j,i)\le s\}, \ s\ge 0$.
The  above diagonal matrix  ${\bf P}_{\bf H}$ 
 can be evaluated at  vertex level, and   constructed on SDNs with communication range $L\ge \omega({\bf H})$,
  see
 Algorithm \ref{preconditioningmatrix.algorithm}.
 \begin{algorithm}[t]
\caption{Realization of the preconditioner 
 ${\bf P}_{\bf H}$ at a vertex $i\in V$. }
\label{preconditioningmatrix.algorithm}
\begin{algorithmic}  

\STATE {\bf Inputs}:  Geodesic width $\omega({\bf H})$ of the filter ${\bf H}$ and
nonzero entries
 $H(i,j)$ and $H(j, i)$ for $ j \in B(i,\omega({\bf H}))$ in the $i$-th row and column
 of the filter ${\bf H}$.



\STATE{\bf 1)} Calculate\\
 $d(i)=\max\Big\{\sum\limits_{j\in B(i,\omega({\bf H}))}|H(i,j)|,\sum\limits_{j\in B(i,\omega({\bf H}))}|H(j,i)|\Big\}$.
\STATE{\bf 2)} Send $d(i)$ to all neighbors $k\in B(i,\omega({\bf H}))\backslash \{i\}$ and receive $d(k)$ from neighbors $k\in B(i,\omega({\bf H}))\backslash \{i\}$.
\STATE{\bf 3)} Calculate $P_{\bf H}(i, i)=\max\limits_{k\in B(i,\omega({\bf H}))} d(k) $.
\STATE {\bf Output}: $P_{\bf H}(i, i)$.
\end{algorithmic}\vspace{-.03in}
\end{algorithm}

For symmetric  matrices ${\bf A}$ and ${\bf B}$, we use ${\bf B}\preceq  {\bf A}$ and ${\bf B}\prec  {\bf A}$  to denote  the positive semidefiniteness and positive definiteness of their difference ${\bf A}-{\bf B}$ respectively.
A crucial observation about the diagonal matrix 
 $\bf P_H$ is as follows.

 \begin{theorem}\label{hht.tm}
 {\rm Let ${\bf H}=(H(i,j))_{i,j\in V}$
  be a graph filter with  geodesic-width $\omega({\bf H})$
 and
 ${\bf P}_{\bf H}$ be as in \eqref{da.def}. Then
 \vspace{-.3em}\begin{equation}\label{hht.ph}
{\bf H}^T {\bf H}\preceq  {\bf P}_{\bf H}^2.
\vspace{-.3em}\end{equation}
}
\end{theorem}

\begin{proof} 
For ${\bf x}=(x(i))_{i\in V}$, we have
\vspace{-.5em}
\begin{eqnarray*}
  0 & \hskip-0.08in \le & \hskip-0.08in  {\bf x}^T {\bf H}^T{\bf H} {\bf x}  
   = \sum_{j\in V} \Big| 
   \sum_{i\in V} 
   H(j,i)x(i) \Big|^2\\
  & \hskip-0.08in \leq & \hskip-0.08in\sum_{j\in V}\Big(\sum_{ i\in V}|H(j,i)||x(i)|^2\Big)\times \Big( \sum_{ i'\in V}|H(j,i')|\Big)\\
   & \hskip-0.08in = & \hskip-0.08in\sum_{i \in V} |x(i)|^2  \sum_{j\in  B(i, \omega({\bf H})) }   |H(j,i)| \times \Big( \sum_{ i'\in V}|H(j,i')|\Big)     \nonumber\\
   & \hskip-0.08in \leq & \hskip-0.08in \sum_{i \in V} |x(i)|^2  P_{\bf H}(i,i) 
    \sum_{j\in  B(i, \omega({\bf H})) } |H(j,i)|
     \nonumber\\
   & \hskip-0.08in \le  & \hskip-0.08in \sum_{i\in V}  ({ P}_{\bf H}(i, i))^2 |x(i)|^2={\bf x}^T  {\bf P}_{\bf H}^2 {\bf x}. 
\vspace{-.7em}
\end{eqnarray*}
This  proves \eqref{hht.ph} and  completes the proof.
\end{proof}

Denote the spectral radius of a matrix ${\bf A}$ by $r({\bf A})$.
By
Theorem  \ref{hht.tm},  ${\bf P}_{\bf H}^{-2} {\bf H}^T$ is an approximation  filter to the inverse filter ${\bf H}^{-1}$ in the sense that 
\vspace{-0.3em}
\begin{equation}\label{hht.singularvalue}
r( {\bf I}-{\bf P}_{\bf H}^{-2} {\bf H}^T
{\bf H})=
r( {\bf I}-{\bf P}_{\bf H}^{-1} {\bf H}^T
{\bf H}{\bf P}_{\bf H}^{-1})<1.
\end{equation}

\begin{remark}{\rm
Define the Schur norm of a matrix  ${\bf H}=(H(i,j))_{i,j\in V}$  by
\vspace{-.5em}\begin{equation*}\label{schurnorm.def}
\hskip-.04in \|{\bf H}\|_{\mathcal S}   =  \max \Big\{ \max_{i\in V} \sum_{j\in V} |H(i,j)|, \
\max_{j\in V} \sum_{i\in V} |H(i, j)|\Big\}, 
\vspace{-.5em}
\end{equation*}
and denote  the zero and identity matrices of appropriate size by ${\bf O}$ and  ${\bf I}$ respectively.
One may verify that
\vspace{-.3em}
\begin{equation}\label{hht.schur}
{\bf O}\prec {\bf H}^T {\bf H} \preceq\|{\bf H}\|_{\mathcal S}^2 {\bf I}.
\vspace{-.5em}
\end{equation}
By \eqref{da.def}, 
we have 
\vspace{-.7em}
\begin{equation}\label{da.schur}
{\bf P}_{\bf H}\preceq\|{\bf H}\|_{\mathcal S} {\bf I}. \vspace{-.5em}\end{equation}
Then we may consider the conclusion
\eqref{hht.ph} for the  preconditioner ${\bf P}_{\bf H}$ 
as a distributed version of the well-known matrix dominance  \eqref{hht.schur} for the graph filter ${\bf H}$.
}
\end{remark}

 \begin{algorithm}[t]
\caption{Implementation of the PGDA \eqref{gradientdescent.al2.default} at a vertex $i\in V$. }
\label{distributed_ICPA.algorithm}
\begin{algorithmic}  

\STATE {\bf Inputs}:
Iteration number $M$,
geodesic-width $\omega({\bf H})$,
preconditioning constant $P_{\bf H}(i, i)$, observation $y(i)$ at vertex $i$,
 and
  filter  coefficients $H(i,j)$ and $H(j, i), \ j \in B(i,\omega({\bf H}))$.

\STATE{\bf 1)} Calculate $\widetilde H(j,i)= H(j,i)/(P_{\bf H}(i, i))^2$.


\STATE {\bf Initialization}:  Initial $x^{(0)}(j), j\in B(i, \omega({\bf H}))$, and $m=1$.

\STATE{\bf 2)} Calculate $v^{(m)}(i)= y(i)- \sum\limits_{j\in B(i,\omega({\bf H}))}H(i,j) x^{(m-1)}(j)$.
\STATE{\bf 3)} Send $v^{(m)}(i)$ to neighbors $j\in B(i,\omega({\bf H}))$
and receive $v^{(m)}(j)$ from neighbors $j\in B(i,\omega({\bf H}))$.

\STATE{\bf 4)} Update \\
 $x^{(m)}(i)=x^{(m-1)}(i) + \sum\limits_{j\in B(i, \omega({\bf H}))}\widetilde H(j, i) v^{(m)}(j).$

 \STATE{\bf 5)}   Send $x^{(m)}(i)$ to neighbors $j\in B(i,\omega({\bf H}))$ and receive $x^{(m)}(j)$ from neighbors $j\in B(i,\omega({\bf H}))$.

\STATE{\bf 6)}
 Set $m=m+1$ and return to  Step   {\bf 2)} if $m\le M$.

\STATE {\bf Outputs}:  $x(j):=x^{(M)}(j), \ j\in B(i, \omega({\bf H}))$.  

\end{algorithmic}\vspace{-.03in}
\end{algorithm}

Preconditioning technique has been widely used in numerical analysis to solve a linear system, where the difficulty is how to select the preconditioner appropriately.
In this letter, we  use  ${\bf P}_{\bf H}$   as a right preconditioner to
  the linear system
\vspace{-.5em}\begin{equation}\label{inversefiltering.eq}
{\bf H}{\bf x}={\bf y}
\vspace{-.5em}\end{equation}
associated with the inverse filtering procedure \eqref{inversefilterprocedure.def},
and we solve
 the following right preconditioned linear system
\vspace{-.5em}\begin{equation}\label{Preconditioningsystem}
{\bf H} {\bf P}_{\bf H}^{-1} {\bf z}={\bf y}\ \ {\rm and}\ \ {\bf x}= {\bf P}_{\bf H}^{-1} {\bf z},
\vspace{-.5em}\end{equation}
via  
 the  gradient descent algorithm
  	\vspace{-.3em}	\begin{equation*}\label{gradientdescent.al}
	\left\{\begin{array}{l}		{\bf z}^{(m)}   =   {\bf z}^{(m-1)}- {\bf P}_{\bf H}^{-1} {\bf H}^T
\big( {\bf H}{\bf P}_{\bf H}^{-1} {\bf z}^{(m-1)}- {\bf y}\big)\\ 
{\bf x}^{(m)}=  {\bf P}_{\bf H}^{-1} {\bf z}^{(m)}, \
\ m\ge 1,\end{array}\right.
\vspace{-.3em}
			\end{equation*}
with initial ${\bf z}^{(0)}$.  The above iterative algorithm 
 can be reformulated as a quasi-Newton method \eqref{Approximationalgorithm} with ${\bf G}$ replaced by ${\bf P}_{\bf H}^{-2}{\bf H}^T$,
	\vspace{-0.3em}	\begin{equation}\label{gradientdescent.al2.default}
	\left\{ \begin{array}{l}
	{\bf e}^{(m)}= {\bf H}{\bf x}^{(m-1)}-{\bf y}\\
	{\bf x}^{(m)}={\bf x}^{(m-1)}-{\bf P}_{\bf H}^{-2}{\bf H}^T{\bf e}^{(m)}, \ \ m\ge 1
	\end{array} \right.
	\vspace{-0.3em}		\end{equation}
with initial ${\bf x}^{(0)}$.
We call the above approach to implement the inverse filtering procedure \eqref{inversefilterprocedure.def} 
 by the  {\em  preconditioned gradient descent algorithm}, or PGDA for  abbreviation.

Define
${\bf w}_m:={\bf P}_{\bf H}({\bf x}^{(m)}-{\bf H}^{-1} {\bf y}), \ m\ge 0$,
and  the norm $\|{\bf x}\|_2=(\sum_{j\in V} |x(j)|^2)^{1/2}$ for  ${\bf x}=(x_j)_{j\in V}$.
By \eqref{gradientdescent.al2.default}, 
we have
\vspace{-.3em}\begin{equation} \label{wm.def01}
{\bf w}_m=
\big({\bf I}- {\bf P}_{\bf H}^{-1} {\bf H}^T {\bf H} {\bf P}_{\bf H}^{-1}\big)
{\bf w}_{m-1}, \ m\ge 1.  
\vspace{-.3em}\end{equation}
Therefore  the iterative  algorithm \eqref{gradientdescent.al2.default} 
converges exponentially by \eqref{hht.singularvalue} and \eqref{wm.def01}.

 \begin{theorem}\label{exponentialconvergence.tm}
{ \rm Let ${\bf H}=(H(i,j))_{i,j\in V}$ be an invertible graph filter 
 and  ${\bf x}^{(m)}, m\ge 0$, be as in \eqref{gradientdescent.al2.default}. 
 Then
\vspace{-0.4em}
\begin{eqnarray*}
\|{\bf P}_{\bf H}({\bf x}^{(m)}-{\bf H}^{-1} {\bf y})\|_2 & \hskip-0.08in  \le & \hskip-0.08in
\big(r( {\bf I}-{\bf P}_{\bf H}^{-1} {\bf H}^T
{\bf H}{\bf P}_{\bf H}^{-1})\big)^m \nonumber\\
   \hskip-0.08in &
 &  \hskip-0.08in
 \times \|{\bf P}_{\bf H}({\bf x}^{(0)}-{\bf H}^{-1} {\bf y})\|_2, \ m\ge 0.
\end{eqnarray*}}
\end{theorem}

 In addition to the exponential convergence in Theorem \ref{exponentialconvergence.tm},
 the PGDA is that each iteration   
can be implemented at vertex level, see  Algorithm
\ref{distributed_ICPA.algorithm}. Therefore
for an invertible filter ${\bf H}$ with $\omega({\bf H})\le L$,
  the PGDA \eqref{gradientdescent.al2.default} can implement the  inverse filtering procedure
\eqref{inversefilterprocedure.def} on  SDNs  with each agent only storing, computing and exchanging  the information  in a  $L$-hop neighborhood.

\vspace{-.06in}

\section{Symmetric preconditioned gradient descent algorithm for inverse filtering}\label{symmetricPGDA.section}


 \begin{algorithm}[t]
\caption{Implementation of the SPGDA \eqref{Sym_neumann.Alg} at a vertex $i\in V$. }
\label{SPGDA.algorithm}
\begin{algorithmic}  

\STATE {\bf Inputs}:
Iteration number $M$,  geodesic-width $\omega({\bf H})$,
observation $y(i)$ at vertex $i$,
 and   filter  coefficients $H(i,j)$ and $H(j, i),  j \in B(i,\omega({\bf H}))$.



\STATE {\bf 1)} Calculate $P_{\bf H}^{\rm sym}(i,i)=\sum\limits_{j\in B(i, \omega({\bf H}))} |H(i,j)|$,
  $\widetilde H(i,j)={ H(i, j)}/{P_{\bf H}^{\rm sym}(i, i)} $ and $\tilde y(i)={y(i)}/{P_{\bf H}^{\rm sym}(i, i)}, j\in B(i, \omega({\bf H})) $.

\STATE {\bf Initialization}:  Initial $x^{(0)}(j), j\in B(i, \omega({\bf H}))$ and $m=1$.

\STATE{\bf 2)} Compute\\
 $x^{(m)}(i)=x^{(m-1)}(i) + \tilde y(i)- \sum\limits_{j\in B(i, \omega({\bf H}))}\widetilde H(i, j) x^{(m-1)}(j)$.

 \STATE{\bf 3)}   Send $x^{(m)}(i)$ to neighbors $j\in B(i,\omega({\bf H}))$ 
and  receive $x^{(m)}(j)$ from neighbors $j\in B(i,\omega({\bf H}))$.

\STATE{\bf 4)}
 Set $m=m+1$ and return to  Step   {\bf 2)} if $m\le M$.

\STATE {\bf Outputs}:  $x(j):=x^{(M)}(j), j\in B(i,\omega({\bf H}))$.  

\end{algorithmic}
\vspace{-.03in}
\end{algorithm}

In this section, we consider implementing the inverse filtering procedure \eqref{inversefilterprocedure.def} associated with
a {\bf positive definite}  
filter   ${\bf H}=(H(i,j))_{i,j\in V}$
   on   a   connected, undirected and unweighted graph  ${\mathcal G}$. 
Define the diagonal matrix ${\bf P}_{\bf H}^{\rm sym}$ with
 diagonal entries
 \vspace{-.3em}\begin{equation}\label{DiagonalElementSymmetric}
{P}_{\bf H}^{\rm sym}(i, i)=\sum_{j\in B(i, \omega({\bf H}) )}|H(i,j)|,\  i\in V,
\vspace{-.3em}\end{equation}
and set
\vspace{-.3em}\begin{equation}
\label{Hsym.def}
\widehat{\bf H}=({\bf P}_{\bf H}^{\rm sym})^{-1/2} {\bf H} ({\bf P}_{\bf H}^{\rm sym})^{-1/2}. \vspace{-.3em}\end{equation}
We remark that the normalized matrix  
 in \eqref{Hsym.def} associated with a diffusion matrix 
has been used to understand diffusion process
 \cite{nadler2006},
 and the one
 corresponding to  the \mbox{Laplacian} ${\bf L}_{\mathcal G}
 $ on the graph ${\mathcal G}$
 is  half of its normalized Laplacian ${\bf L}_{\mathcal G}^{\rm sym}:=({\bf D}_{\mathcal G})^{-1/2} {\bf L}_{\mathcal G}({\bf D}_{\mathcal G})^{-1/2}$, where
${\bf D}_{\mathcal G}$ is degree matrix of  ${\mathcal G}$  \cite{jiang19}.
 Similar to the PGDA \eqref{gradientdescent.al2.default},
we propose the following   {\em symmetric preconditioned gradient descent algorithm}, or SPGDA for  abbreviation,
 \vspace{-.3em} \begin{equation}\label{Sym_neumann.Alg}
{\bf x}^{({m})}={\bf  x}^{(m-1)}-({\bf P}^{\rm sym}_{\bf H})^{-1}({\bf H}{\bf  x}^{(m-1)}- {\bf y}),\  m\ge 1,
\vspace{-.3em} \end{equation}
with initial ${\bf x}^{(0)}$,
  to
 solve the following  preconditioned linear system
\vspace{-.3em}\begin{equation}\label{Preconditioningsystem.symmetric}
\widehat{\bf H} {\bf z}=({\bf P}_{\bf H}^{\rm sym})^{-1/2}{\bf y}\ \ {\rm and} \ \
 {\bf x}= ({\bf P}_{\bf H}^{\rm sym})^{-1/2} {\bf z}.
\vspace{-.3em}\end{equation}
Comparing with the PGDA \eqref{gradientdescent.al2.default},
the SPGDA  
for a positive definite
graph filter has less computation and communication cost in each iteration
 and it also can be implemented at vertex level, 
 see Algorithm \ref{SPGDA.algorithm}.

For ${\bf x}=(x(i))_{i\in V}$, we obtain from \eqref{DiagonalElementSymmetric} and
the symmetry of the matrix ${\bf H}$ that
\vspace{-.3em}\begin{equation*}
 {\bf x}^T{\bf H}{\bf x} 
\le   \sum_{i, j\in V}  |H(i,j)|\frac{(x(i))^2+(x(j))^2}{2}=
 {\bf x}^T{\bf P}_{\bf H}^{\rm sym}{\bf x}. 
\vspace{-.3em}\end{equation*}
Combining  \eqref{da.def}  
and \eqref{DiagonalElementSymmetric}  proves that
\vspace{-.3em}\begin{equation} \label{symmetric.remark.eq2}
 {\bf H}\preceq {\bf P}_{\bf H}^{\rm sym} \preceq  {\bf P}_{\bf H}, 
\vspace{-.3em}\end{equation}
cf.  \eqref{hht.ph}. 
This together with  \eqref{Hsym.def}
implies that
\vspace{-.3em}\begin{equation}
r({\bf I}-({\bf P}_{\bf H}^{\rm sym})^{-1} {\bf H})= r({\bf I}- \widehat {\bf H})<1.
\vspace{-.3em}\end{equation}
Similar to the proof of  Theorem
  \ref{exponentialconvergence.tm}, we have
%

   \begin{theorem}\label{symmetricexponentialconvergence.tm}
{\rm Let ${\bf H}$ be a positive definite graph filter. 
Then  ${\bf x}^{(m)}, m\ge 0$,  in \eqref{Sym_neumann.Alg} converges exponentially, 
\vspace{-.3em}\begin{eqnarray*} \label{exponentialcovergence}
  \hskip-0.08in & \hskip-0.08in &\nonumber \|({\bf P}^{\rm sym}_{\bf H})^{1/2}({\bf x}^{(m)}-{\bf H}^{-1} {\bf y})\|_2 \\
  & \hskip-0.08in \le & \hskip-0.08in
  \big(r({\bf I}-({\bf P}_{\bf H}^{\rm sym})^{-1} {\bf H})\big)^m
\big\|({\bf P}^{\rm sym}_{\bf H})^{1/2}({\bf x}^{(0)}-{\bf H}^{-1} {\bf y})\big\|_2. 
\vspace{-.3em}\end{eqnarray*}
}
\end{theorem}

\vspace{-.15in}
\section{Numerical simulations}\label{sec:num}

\vspace{-.01in}

\begin{figure}[t] 
\begin{center}
\includegraphics[width=27mm, height=25mm]{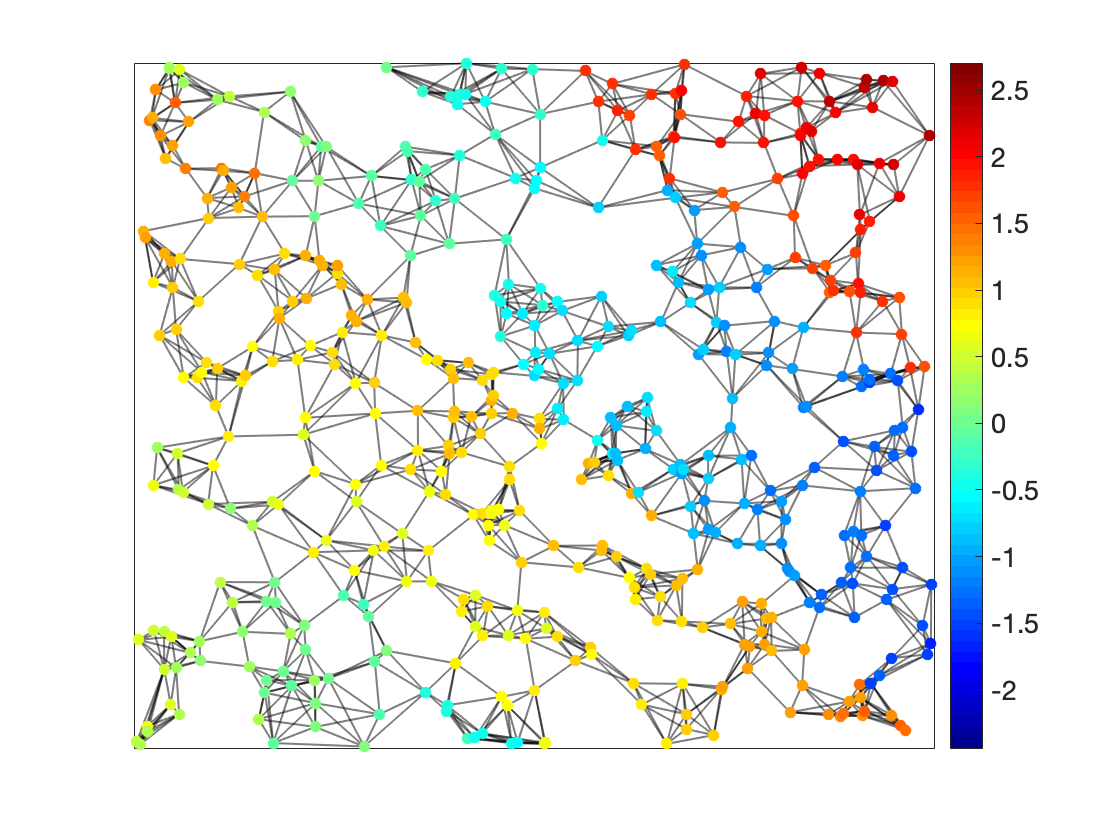}
\includegraphics[width=27mm, height=25mm]{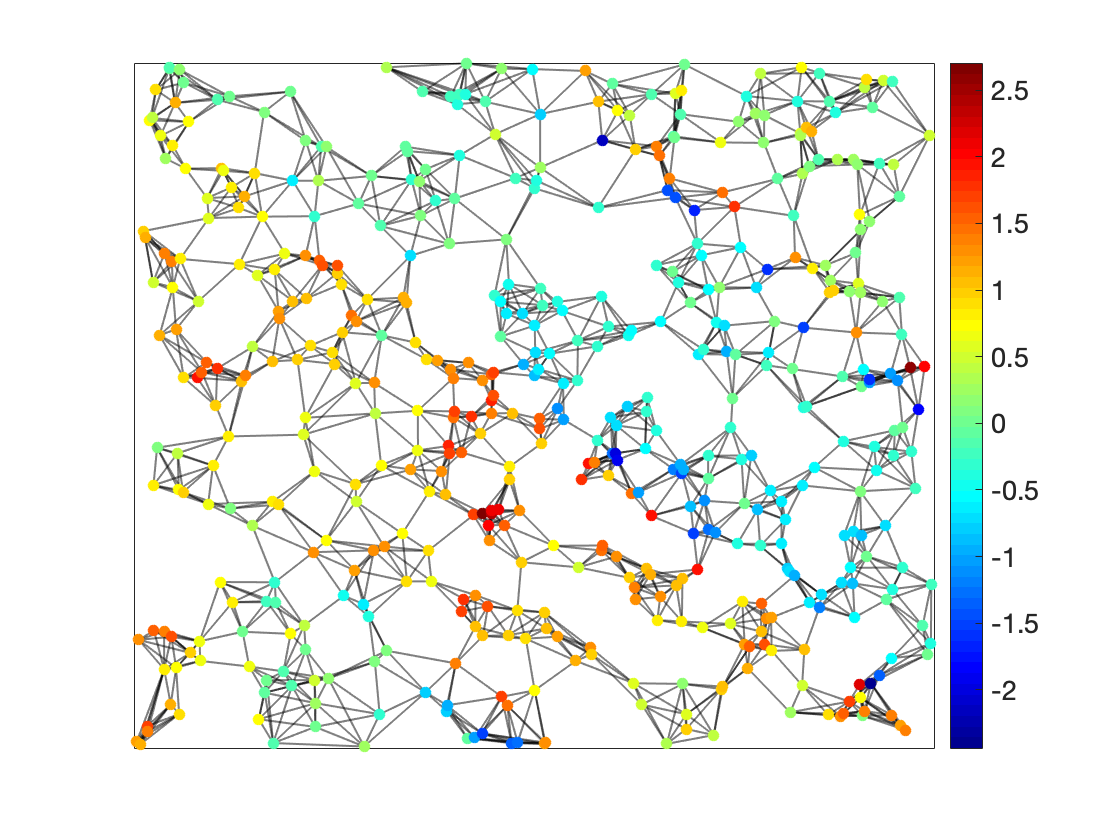}
\includegraphics[width=32mm, height=25mm]{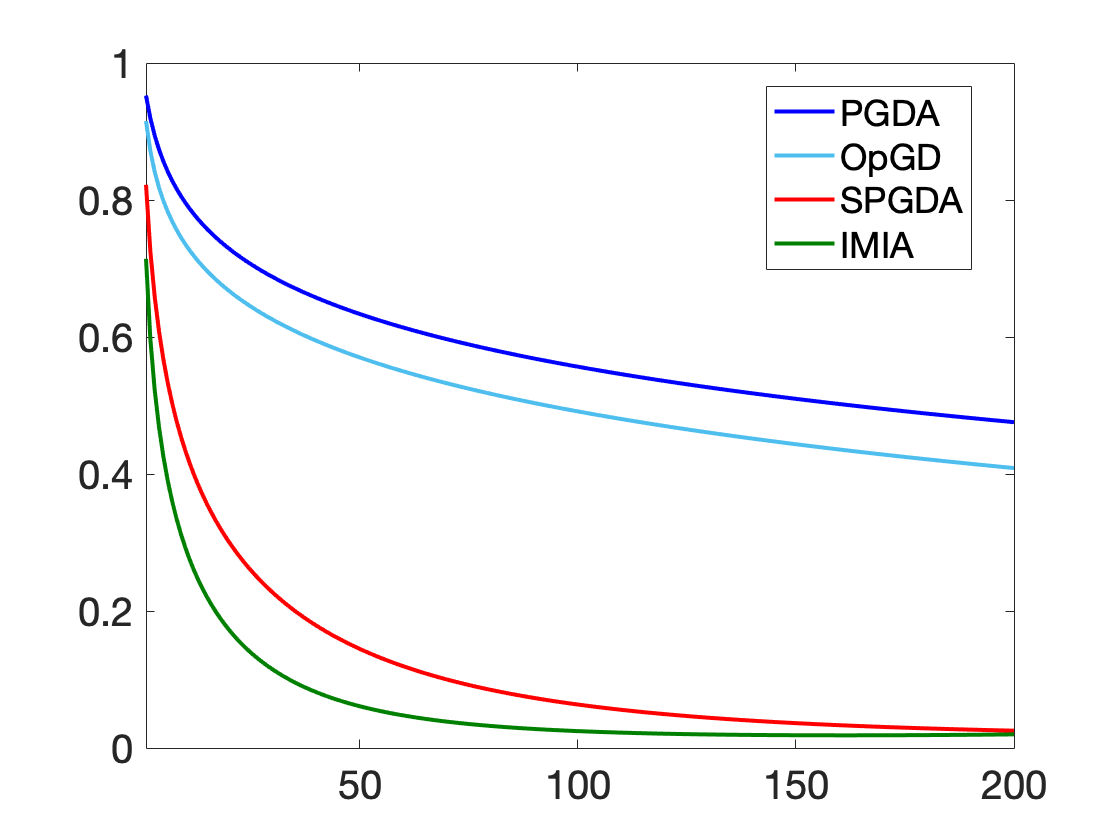}
\caption{\small Plotted on the left  is  a corrupted  blockwise polynomial signal ${\bf x}$
and in the middle is the output   ${\bf y}={\bf H} {\bf x}$ of the filtering procedure,
where $\|{\bf x}\|_2=24.8194, \|{\bf y}\|_2=21.5317$ and  the  condition number of the filter ${\bf H}$ is $107.40$. 
 Shown on the right is average of the  relative inverse filtering error  $E_2(m)= {\|{\bf x}^{(m)}-{\bf x}\|_2}/{\|{\bf x}\|_2}, 1\le m\le  200$
over 1000 trials, where $\eta=0.2, \gamma=0.05$ and ${\bf x}^{(m)}$, $m\ge 1$, are the outputs of  SPGDA, PGDA,  OpGD and IMIA.
}
\label{comp_alg}
\end{center}
\vspace{-.25in} 
\end{figure}

Let  ${\mathcal G}_{512}=(V_{512}, E_{512})$ be a random geometric graph with 512 vertices 
deployed on  
 $[0, 1]^2$ and
  an undirected edge between two vertices if their physical
distance is not larger than
$\sqrt{2/512} 
$  \cite{jiang19, Nathanael2014}.
  In the first simulation, we consider
  the inverse filtering procedure associated with the graph filter
\vspace{-.8em}
\begin{equation*}{\bf H}= {\bf H}_o+ ({\bf L}_{{\mathcal G}_{512}}^{\rm sym})^2,
 \vspace{-.4em}\end{equation*}
 where  $ {\bf L}_{{\mathcal G}_{512}}^{\rm sym}$ is   the
  normalized Laplacian on the graph ${\mathcal G}_{512}$,
 the filter ${\bf H}_o=(H_o(i, j))_{i, j\in V_{512}}$  is defined by  
\vspace{-.3em}\begin{eqnarray*} \hskip-0.08in &  \hskip-0.08in &  \hskip-0.08in H_o(i, j)  =   \exp\Big(-2\times 512 \times \|(i_x, i_y)-(j_x, j_y)\|_2^2-\\
 \hskip-0.08in &   \hskip-0.08in  & \quad \quad -\frac{\|(i_x, i_y)+(j_x, j_y)\|_2^2}{2} \Big)+\frac{\gamma_{ij}+\gamma_{ji}}{2}, \ \rho(i,j)\le 2,
\vspace{-.3em}\end{eqnarray*}
$(i_x, i_y)$ is the coordinator of a vertex $i\in V_{512}$ and
$\gamma_{ij}$ are i.i.d random noises uniformly distributed on $[-\gamma, \gamma]$.
 Let  ${\bf x}_o$  be the blockwise polynomial consisting of four strips and imposes  $(0.5-2 i_x)$ on the first and third diagonal strips and $(0.5 + i_x^2 + i_y^2)$ on the second and fourth strips respectively \cite{jiang19, Emirov20}.
In the  simulation, the signals
\vspace{-.6em}\begin{equation*}{\bf x}={\bf x}_o+\pmb\eta
\vspace{-.4em}\end{equation*}
 are obtained by  a blockwise polynomial ${\bf x}_o$ corrupted by  noises $\pmb \eta$ with their components being i.i.d. random variables with uniform distribution on $[-\eta, \eta]$,
and  the observations ${\bf y}$ of the filtering procedure are given by ${\bf y}={\bf H}{\bf x}$,
 see the left and middle images of
Figure \ref{comp_alg}.
In the  simulation, we use  the SPGDA  \eqref{Sym_neumann.Alg} and the PGDA \eqref{gradientdescent.al2.default}
with zero initial
		to implement  the inverse filtering procedure ${\bf y}\mapsto {\bf H}^{-1}{\bf y}$, 
and also we compare their performances with the 
gradient decent algorithm   
\vspace{-.3em} \begin{equation}\label{gd.algorithm}
{\bf x}^{(m)}= ({\bf I} -\beta_{op} {\bf H}^T{\bf H}){\bf x}^{(m-1)}+\beta_{op}{\bf H}^T {\bf y}, \ m\ge 1
\vspace{-.3em}\end{equation}
 with
zero initial  and
  optimal step length $\beta_{op}$ selected in \cite{ sihengTV15,  Emirov20, Shi15},
 OpGD in abbreviation,
and   the  iterative matrix
inverse approximation algorithm,
\vspace{-.3em} \begin{equation}\label{gd.algorithm}
{\bf x}^{(m)}= ({\bf I} - \tilde {\bf D} {\bf H}){\bf x}^{(m-1)}+
 \tilde {\bf D}  {\bf y}, \ m\ge 1
\vspace{-.3em}\end{equation}
 IMIA in abbreviation, where ${\bf x}^{(0)}={\bf 0}$ and the diagonal matrix $\tilde {\bf D}$ has entries
   $H(i,i)/(\sum_{\rho(j,i)\le 2} |H(i,j)|^2), i\in V$, see \cite[Eq. (3.4)]{Tay19} with $\tilde \sigma=0$.
Shown in   Figure \ref{comp_alg} is
the average of the relative inverse filtering error  $E_2(m), 1\le m\le 200$
over 1000 trials, 
 and it reaches the relative error 5\% at  about  $57$th iteration for IMIA,   118th iteration for
 SPGDA, and more than 3000 iterations for PGDA   and OpGD.
 This  confirms  
 that ${\bf x}^{(m)}, m\ge 1$, in  the
  SPGDA, PGDA,  OpGD and IMIA converge exponentially to the output  ${\bf x}$ of the inverse filtering,
 and  the  convergence rate
 are   spectral radii
 of  matrices  ${\bf I}-({\bf P}_{{\bf H}}^{\rm sym})^{-1} {\bf H}$, ${\bf I}-{\bf P}_{{\bf H}}^{-2}{\bf H}^{T}{\bf H}$,  
  ${\bf I}-\beta_{op}{\bf H}^{T}{\bf H}$ and ${\bf I}-\tilde {\bf D}{\bf H}$, see Theorems \ref{exponentialconvergence.tm} and \ref{symmetricexponentialconvergence.tm}. Here
   the average of spectral radii in SPGDA, PGDA,  OpGD and IMIA are $0.9786, 0.9996,  0.9993,
 0.9566$   respectively. We remark that the reason for PGDA and OpGD to have slow convergence in the above simulation could be that  their spectral radii are too close to $1$.
For the filter perturbation level $\gamma=0.1$, our simulations indicate that for some filters ${\bf H}$ being invertible but not positive definite, the corresponding PGDA and OpGD converge
while  SPGDA and IMIA diverge.

\smallskip

\begin{figure}[t] 
\begin{center}
\includegraphics[width=45mm, height=25mm]{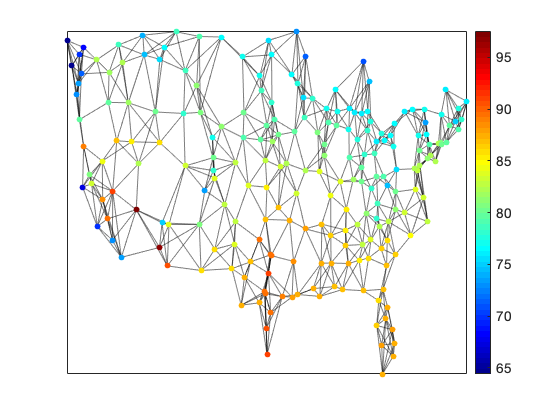}
\includegraphics[width=33mm, height=25mm]{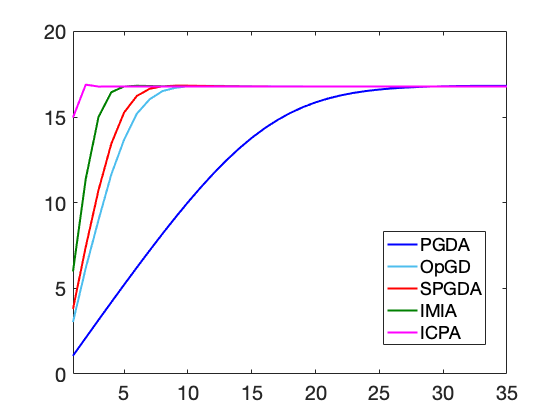}
\caption{\small Plotted on the left is  the original temperature data ${\bf x}_{12}$.
 Shown on the right is average of the signal-to-noise ratio  $\text{SNR}(m)=-20\log_{10}{\|{\bf x}^{(m)}-{\bf x}_{12}\|_2}/{\|{\bf x}_{12}\|_2}, 1\le m\le 35$,
over 1000 trials, where ${\bf x}^{(m)}$, $m\ge 1$, are the outputs of PGDA,  SPGDA, OpGD, IMIA and ICPA, and average of the limit  SNR
is 16.7869.
}
\label{US_denoising12}
\end{center}
\vspace{-.25in} 
\end{figure}
Let ${\mathcal G}_T=(V_T,E_T)$ be the undirected graph with 218 locations in  the  United  States as vertices and edges constructed by the 5 nearest neighboring locations,
and let  ${\bf x}_{12}$ be the recorded temperature vector of those 218 locations on August 1st, 2010 at 12:00 PM, see  Figure \ref{US_denoising12}
\cite{Emirov20, zeng17}.
In the  second simulation, we consider to implement the inverse filtering procedure
\vspace{-.4em}
\begin{equation*}    \tilde{\bf x}=({\bf I}+\alpha {\bf L}_{\mathcal{G}_T}^{\rm sym})^{-1}{\bf b}
\vspace{-.3em}\end{equation*}
arisen from the minimization problem
\vspace{-.3em}
\begin{equation*}   \tilde{\bf x}:=\arg\min_{{\bf z}}\|{\bf z}-{\bf b}\|_{2}^2+\alpha {\bf z}^T{\bf L}_{\mathcal{G}_T}^{\rm sym}{\bf z}
\vspace{-.3em}\end{equation*}
 in denoising the hourly temperature data ${\bf x}_{12}$,
where ${\bf L}_{\mathcal{G}_T}^{\rm sym}$ is the normalized Laplacian on ${\mathcal G}_T$,
 $\alpha$ is a penalty constraint  and ${\bf b}={\bf x}_{12}+\pmb \eta$ is the temperature vector corrupted by i.i.d. random noise $\pmb \eta$ with its components being
 randomly  selected  in $[-\eta,\eta]$ in a uniform distribution \cite{Emirov20, zeng17}.
Shown in  Figure \ref{US_denoising12} is the performance of the SPGDA, PGDA, OpGD, IMIA and ICPA to implement the above
inverse filtering procedure with noise level $\eta=35$ and the penalty constraint $\alpha=0.9075$ \cite{Emirov20},
where ICPA is the iterative Chebyshev polynomial approximation algorithm of order one   \cite{Shuman18, Emirov19, Emirov20}.
 This indicates that the 3rd term in ICPA, the 5th term in IMIA, the 8th term of SPGDA, the 10th term of OpGD and the 30th term of PGDA
 can be used as the denoised  temperature vector $\tilde {\bf x}$. 

To implement the inverse filter procedure \eqref{inversefilterprocedure.def} on SDNs,  we observe from the above two simulations that
OpGD  outperforms  PGDA while the selection of optimal step length in OpGD
is computationally expensive. If the filter is positive definite,   SPGDA, IMIA and ICPA may have better performance 
than OpGD and PGDA have.
On the other hand, SPGDA always converges, 
but  the requirement in \cite[Theorem 3.2]{Tay19}
to guarantee the convergence of IMIA may not be  satisfied and ICPA is applicable  for polynomial filters.

\end{document}